\documentclass{sig-alternate}

\usepackage{algorithmic}
\usepackage{algorithm}
\usepackage{subfigure}
\usepackage{epsfig}
\usepackage{multirow}
\usepackage{amsmath}
\usepackage{amsfonts}
\usepackage{array}
\usepackage{wrapfig}

\newtheorem{theorem}{Theorem}[section]
\newtheorem{lemma}[theorem]{Lemma}

\newtheorem{corollary}[theorem]{Corollary}
			
\begin{document}

\numberofauthors{2}

\title{\Large \bf Sequential Hypothesis Tests for Adaptive Locality Sensitive Hashing}

\author{
\alignauthor
Aniket Chakrabarti\\
       \affaddr{The Ohio State University}\\
       \email{chakrabarti.14@osu.edu}
\alignauthor
Srinivasan Parthasarathy\\
       \affaddr{The Ohio State University}\\
       \email{srini@cse.ohio-state.edu}
			}

\maketitle

\begin{abstract}
All pairs similarity search is a problem where a set of data
objects is given and the task is to find all pairs of objects
that have similarity above a certain threshold for a given 
similarity measure-of-interest. When the number of points or dimensionality
is high, standard solutions fail to scale gracefully. Approximate solutions
such as Locality Sensitive Hashing (LSH) 
and its Bayesian variants (BayesLSH and BayesLSHLite) alleviate the problem to some extent
and provides substantial speedup over
traditional index based approaches. BayesLSH is used for pruning the
candidate space and computation of approximate similarity, whereas
BayesLSHLite can only prune the candidates, but similarity needs to be
computed exactly on the original data. Thus where ever the explicit
data representation is available and exact similarity computation is
not too expensive, BayesLSHLite can be used to aggressively prune
candidates and provide substantial speedup without losing too much
on quality. However, the loss in quality is higher
in the BayesLSH variant, where explicit data representation
is not available, rather only a hash sketch is available and
similarity has to be estimated approximately.
In this work we revisit the LSH problem from a Frequentist setting and formulate
sequential tests for composite hypothesis (similarity greater than or less than
threshold)
that can be leveraged by such LSH algorithms for adaptively
pruning candidates aggressively. We propose a vanilla sequential probability ration test (SPRT) approach based on this idea and 
two novel variants. We extend these variants to the case where
approximate similarity needs to be computed using fixed-width
sequential confidence interval generation technique.
We compare these novel variants with 
the SPRT variant and BayesLSH/BayesLSHLite variants and show
that they can provide tighter qualitative guarantees over BayesLSH/BayesLSHLite -- a state-of-the-art approach -- 
while being upto 2.1x faster than a traditional SPRT and 8.8x
faster than AllPairs.
\end{abstract}

\section{Introduction}
\label{sec:intro}
Similarity search in a collection of objects has a wide variety
of applications such as clustering~\cite{ravichandran05}, semi-supervised
learning~\cite{zhu09}, information retrieval~\cite{frakes1992information}, query refinement on
websearch~\cite{bayardo07}, near duplicate detection~\cite{xiao11}, collaborative filtering,
and link prediction~\cite{liben-nowell07}. Formally, the problem
statement is:
\textit{ Given a collection of objects D and an associated
similarity measure $s(.,.)$ and a similarity threshold $t$,
 the problem is to find all pairs
of objects $x,y$, such that $s(x,y) \geq t$, where $x,y\in D$}.

The
major challenge in most of these areas is dealing with a large
volume of the data. The volume combined with the high dimensionality
of the datasets can lead to inefficient solutions to this problem.
Recent research has focused
on reducing the candidate search space. The AllPairs~\cite{bayardo07}
candidate generation algorithm builds a smart index
structure from the vector representation of the points to
give the exact similarity. Another approach is to do an approximate
similarity calculation which involves sacrificing small accuracy for substantial speedup. The most popular
technique is the locality sensitive hashing(LSH)~\cite{indyk98,gionis99} which
involves projecting the high dimensional data into a lower
dimensional space and similarity of a pair of points
is approximated by the number of matching
attributes in the low dimensional space. 

A recent idea, BayesLSHLite technique~\cite{satuluri2012bayesian},
further aggressively prunes the candidate space to generate the set of pairs above the
user-defined similarity threshold efficiently. The way BayesLSHLite
works is, for a pair of data objects $x,y$, BayesLSHLite incrementally
compares their hashes (in batches of size $b$) and infers after each
batch, how likely is it that pair will have similarity above the threshold.
As the name suggests it relies on Bayesian principles and priors for making the 
adaptive decision to prune or retain a candidate.
If that probability becomes too low, then candidate pair is pruned away.
The above algorithm has a variant called BayesLSH, where after the candidate
pruning is done, similarity is estimated approximately from the
hash signatures instead of
exact computation. This is useful in cases where exact similarity
computation is infeasible. Exact similarity computation might be
infeasible in cases where the original data was too large too
store and the small hash sketch of the data was stored instead.
Another scenario could be where the similarity measure-of-interest
is a kernel function and exact representation
of the data points in the kernel induced feature space is not possible
as that space might be infinite dimensional (e.g. Gaussian RBF kernel).
Additionally, the specialized kernel functions are extremely expensive to
compute on the fly. In such cases, both candidate generation and
similarity estimation has to be done approximately using the
LSH hash sketches.

In this paper we adapt a Frequentist view of this idea
and propose a fully principled sequential model to do
the incremental pruning task with rigorous quality guarantees.
Specifically, we model the problem of $s(x,y) \geq t$ as a sequential
hypothesis test problem and provide quality guarantees through Type I and
Type II errors. We start with the traditional SPRT~\cite{wald1973sequential}, and show that
it is extremely inefficient in practice, making it unsuitable for
large scale analytics. We then propose an alternate sequential
hypothesis test procedure based on a one-sided fixed-width
confidence limit construction technique.
Finally we show that no single hypothesis testing strategy
works well for all similarity values. Therefore, we propose
a fine-grained hybrid hypothesis testing strategy, which based on a crude
estimate of the similarity from the first batch of hash
comparisons for that specific candidate pair, selects the most suitable test for that pair.
{\it In other words, instead of using a single hypothesis test,
we dynamically choose the best suited hypothesis test for each
candidate pair.}
We extend the above to technique to develop a variant, that
after candidate pruning, estimates the approximate similarity
by using the fixed-width two-sided confidence interval
generation technique~\cite{frey2010fixed}.
Our ideas are simple to implement
and we show that our hybrid method always guarantees the minimum
quality requirement (as specified by parameters input to the algorithm), while being upto 2.1x faster than an
SPRT-based approach and 8.8x faster than AllPairs
 while qualitatively improving on the state-of-the-art BayesLSH/Lite estimates. 

\section{Background}
\subsection{Locality Sensitive Hashing}
Locality sensitive hashing\cite{indyk98,gionis99} is a popular and fast method for 
candidate generation and approximate similarity computation within
 a large high dimensional dataset. Research has demonstrated how to leverage
the key principles for a host of distance and similarity measures~\cite{indyk98,datar04,broder97,ravichandran05,henzinger06}. Briefly, each data point is represented by a set of hash keys using 
a specific hash family for a given distance or similarity measure ($sim$). 
Such a family hash function is said to have 
the locality sensitive hashing property if:
\begin{equation}
P_{h\in F}(h(x) == h(y)) = sim(x,y)
\label{eq:lsh}
\end{equation}
where $x, y$ are the any two points in the dataset, and $h$ is a randomly selected hash function from within a family $F$ of hash functions.
Consequently, the approximate similarity between the pair can be estimated as:
\begin{equation*}
\hat{s(x,y)} = \frac{1}{n} \sum_{i=1}^n I[h_i(x) == h_i(y))]
\end{equation*}
where $k$ is the total number of hash functions.

\subsection{Candidate generation}
We use the AllPairs~\cite{bayardo07} candidate generation algorithm
when the original data set is available. The AllPairs candidate
generation algorithm is exact, hence all true positives (candidates
with similarity above the threshold) will be present in the
set of candidates generated. The AllPairs algorithm builds an
index from vector representation of the data and instead of building
a full inverted index, AllPairs only indexes the information which
may lead to a pair having similarity greater than the specified threshold.
The AllPairs algorithm can be used
when the original dataset is small enough to be stored entirely
and the similarity of interest is computed on the original
feature space (unlike kernel similarity measures).

In cases where the entire dataset cannot be stored,
rather a small sketch of it is available, AllPairs cannot be used.
Additionally if the similarity is a form of kernel measure and
the feature space of that kernel cannot be explicitely represented,
AllPairs will not work as it relies on the explicit representation
of the data points. However, in both scenarios we can use LSH
to generate a low dimensional sketch of the dataset and use
the following probabilistic candidate generation algorithm.
We follow the general index structure for LSH candidate generation used in \cite{indyk98,datar04,broder97,ravichandran05,henzinger06}. The advantage of such an index structure is that the number of candidates to search for a certain point too see which of them has higher similarity than a threshold $t$ becomes much less compared to exhaustive search. That is search can be be done in sublinear time. LSH based index structures perform well compared to traditional indices when the dimensionality of the data set is very high. The algorithm is as follows:
\textit{
\begin{enumerate}
\item
Using a locality sensitive hashing method for a given similarity measure, form $l$ signatures for each data point, each signature having $k$ hash keys.
\item
Each pair of points, that share at least one signature will stay in the same hash bucket.
\item
During all pairs similarity search, for each point, only those points which are in the same bucket needs to be searched.
\item
From \cite{xiao11}, for a given $k$ and similarity threshold $t$, the number of signatures $l$ required for a recall $1-\phi$, 
\begin{displaymath}
l = \lceil\frac{{log(\phi)}}{log(1-t^k)}\rceil
\end{displaymath}
\end{enumerate}
}

\subsection{Candidate Pruning using BayesLSH/Lite}
Traditionally maximum likelihood estimators are used to 
approximate similarity of a pair to decide whether it 
is above or below a certain threshold.
For a candidate pair, if there are a total of $n$ hashes and $m$ of them match, then the similarity estimate is $\frac{m}{n}$. The variance of this 
estimator is $\frac{s(1-s)}{n}$ where $s$ is the similarity between the candiate pair. Two issues to observe here are - 1) as $n$ 
increases, the variance decreases and hence the accuracy of the estimator increases, 2) more importantly, the variance of the estimator depends on the similarity of the pair itself. This means for a fixed number of hashes, the accuracy achieved if the similarity of the candidate pair was 0.9 is higher than if the similarity was 0.5. In other words the number of hashes required for different 
candidate pairs for achieving the same level of accuracy is different. 
Therefore, the problem with fixing the number of hashes is - some of the candidate pairs can be pruned by comparing the first few hashes only. For example if the similarity threshold is 0.9 and 8 out of the first 10 hashes did not match, there is very low probability that the similarity of the pair is greater than 0.9. The BayesLSH/Lite\cite{satuluri2012bayesian} is among the earliest approaches to solve this
problem of deciding the number of hash comparisons. Instead it 
incrementally compares hashes until the candidate pair can be 
pruned with a certain probability, or the maximum allowed hash comparisons
is reached. It will then compute exact or approximate similarity to make the decision.
To do the incremental pruning, BayesLSHLite solves the inference problem
$P(s(x,y) \ge t)$, where $t$ is the similarity threshold.
Additionally, the BayesLSH variant
estimates the approximate similarity by creating an interval
for the true similarity by the solving the inference
$P(|s(x,y) - \hat{s(x,y)|} \leq \delta)$. Both inferences
are solved using a simple Bayesian model.
These inferences help overcome the second problem pointed above - number of hashes required to
prune a candidate or build an interval around it adaptively set for each candidate pair.

Since we are counting the number of matches $m$ out of $n$ hash comparison, and each comparison is independent with match probability $S$ as per equation~\ref{eq:lsh},
the likelihood function becomes a binomial distribution with parameters $n$ and $S$. If $M(m,n)$ is the random variable denoting $m$ matches out of $n$ hash bit comparisons, and let $S$ be the similarity $s(x,y)$, then the likelihood function will be:
\begin{equation}
P(M(m,n)|S) = {n \choose m}S^m(1-S)^{n-m}
\label{eq:likelihood}
\end{equation}

In the Bayesian setting, the parameter $S$ can be treated as a random variable.
Let the estimate of $S$ in this setting be $\hat{S}$.
Using the aforementioned likelihood function, the two inference problems become:

\textbf{Early pruning inference:} given m matches out of n bits, what is the probability that the similarity is above threshold $t$:
\begin{eqnarray}
P[ S \geq t  \, | \,  M(m,n) ] &=& \int_t^1 P(S  \, | \,  M(m,n)) dS 
\label{eqn:pruneProb}
\end{eqnarray}

\textbf{Concentration inference:} given the similarity estimate $\hat{S}$, what is the probability that it falls within $\delta$ of the true similarity:
\begin{eqnarray}
P[ |S - \hat{S}| < \delta \, | \, M(m,n) ] &= P[ \hat{S}-\delta <
S < \hat{S}+\delta  \, | \,  M(m,n) ] \nonumber \\
&= \int_{\hat{S}-\delta}^{\hat{S}+\delta} P(S  \, | \,  M(m,n)) dS 
\label{eqn:concentrateProb}
\end{eqnarray}

The BayesLSHLite algorithm works as follows:

\noindent \textit{For each pair x,y:
\begin{enumerate}
\item Compare the next $b$ hashes and compute the early pruning probability 
$P[ S \geq t  \, | \,  M(m,n) ]$.
\item If $P[ S \geq t  \, | \,  M(m,n) ] < \alpha$, then prune the pair and stop.
\item If maximum allowable hash comparisons have been reached,
compute exact similarity and stop.
\item Go to step 1.
\end{enumerate}
}

The BayesLSH variant works as follows:

\noindent \textit{For each pair x,y:
\begin{enumerate}
\item Compare the next $b$ hashes and compute the early pruning probability 
$P[ S \geq t  \, | \,  M(m,n) ]$.
\item If $P[ S \geq t  \, | \,  M(m,n) ] < \alpha$, then prune the pair and stop.
\item If $P[ |S - \hat{S}| < \delta \, | \, M(m,n) ] > 1 - \gamma$, then output pair $x,y$
if $\hat{S} \geq t$ and stop.
\item If maximum allowable hash comparisons have been reached,
then output pair $x,y$
if $\hat{S} \geq t$ and stop.
\item Go to step 1.
\end{enumerate}
}

\section{Case for Frequentist Formulation}
\label{sec:motivation}
The BayesLSHLite candidate pruning algorithm
and the BayesLSH approximate similarity estimation
algorithm
 as described in the previous section, provides the basis for the current 
work. Specifically in this work we examine the same problems they attempt
to solve but in a Frequentist setting. We note that the
inferences (equations~\ref{eqn:pruneProb} and~\ref{eqn:concentrateProb}) 
in the above BayesLSH/Lite algorithms
is done every $b$ hash comparisons (this can be viewed as a 
bin of comparisons). Therefore, for a 
candidate pair, if the pruning inference is done once,
then the error probability rate will $\alpha$, but when
it is done for the second time, probability of the pair
getting pruned will be determined by getting pruned the
first time (first bin) and the probability of getting pruned the
second time (a cumulative of the first and second bin matches). 
Essentially, we argue in this
work that this error rate may propogate 
resulting in an accumulated error over multiple pruning inferences.
{\it The underlying reason is, BayesLSHLite tries to model an inherently 
sequential decision process in a non-sequential way.}
The same scenario is true for the concentration inference
as well (equation~\ref{eqn:concentrateProb}). Over multiple
concentration inferences done incrementally, the coverage
probability could fall below $1-\gamma$.
We note that in practice
this may not be a significant issue but the question remains can this
problem be fixed (in the rare cases it may materialize)
without significantly impacting the gains
obtained by BayesLSH/Lite.  We note that fixing this problem in a Bayesian
setting remains open but in this work we show how this problem can
be fixed in a Frequentist setting.

Another issue, again a minor one, is
when a pair is unlikely to be above a certain
similarity threshold, pruning it early saves hash comparisons,
similarly when a pair is very likely to be above the threshold,
hash comparison for it should stop immediately and it should be
processed for exact similarity computation. This can also
save a number of hash comparisons. 

To overcome these 
problems, we propose to model the problem in a Frequentist setting as follows.
In the frequentist setting let the similarity $s(x,y)$ be denoted by the
parameter $s$ (intead of S as in Bayesian setting).
\begin{itemize}
\item
We model the early pruning inference
$s > t$ as a {\it sequential
hypothesis test} problem that should be able to guarantee Type I and Type II
errors under the sequential hash comparison setting and if possible,
it should be able to early prune a pair or send a pair for
exact similarity computation. 
\item
We model the concentration inference $|s-\hat{s}| \leq \delta$ as a
{\it sequential two-sided fixed-width confidence interval} creation
problem that should be able to guarantee a certain coverage
probability.
\end{itemize}

\section{Methodology}
\label{sec:mlest}
In this section, we
describe a principled way of doing the early pruning
inference (equation~\ref{eqn:pruneProb}) and a principled
way of doing the concentration inference (equation~\ref{eqn:concentrateProb})
under the sequential setting where the number of hash functions ($n$)
is not fixed, rather it is also a random variable.
\subsection{Early Pruning Inference}
We use sequential tests of
composite hypothesis for pruning the 
number of the generated candidates, so that the cardinality
of the remaining candidate set is very small. Therefore,
exact similarity computation on the remaining set of candidate
pairs becomes feasible in terms of execution time, provided the
original data set is available and the similarity function can be
computed on the feature space of the original data. Our pruning
algorithm involves sequentially comparing the hashes for a 
pair of data objects and stop when we are able to infer 
with some certainty whether the similarity for the
pair is above or below the user defined threshold. If, according
to the inference,
the similarity
of the pair is below the threshold, then we prune away the pair,
otherwise we compute exact or approximate similarity of the pair
depending on which variant we are using. More formally,
if the similarity of the pair is $s$ and the user defined
threshold is $t$, we need to solve the hypothesis
test, where the null hypothesis is $H0:s\geq t$ and the alternate
hypothesis is $H1:s < t$. Two interesting aspects of our
problem formulation are:
\begin{enumerate}
\item For performance reasons as described in section~\ref{sec:motivation},
we do not want to fix the number of hashes to compare for the
hypothesis test, but rather incrementally compare the hashes
and stop when a certain accuracy in terms of Type I error
has been achieved. 
\item We focus on Type I error, i.e. we do not
want to prune away candidate pairs which are true positives
($s\geq t$). We can allow false positives ($s<t$) in our
final set, as either exact similarity computation or approximate
similarity estimation will be done on the
final set of candidate pairs and any false positives can
thus be pruned away. In other words, we do not need to
provide guarantees on Type II error of the hypothesis test.
Ofcourse keeping a low Type II error implies less false
positives to process, resulting in better performance.
\end{enumerate}

We discuss three strategies for formulating the hypothesis test.
First, we cast our problem in a traditional Sequential Probability
Ratio Test (SPRT)~\cite{wald1973sequential} setting and then discuss
the shortcomings of such an approach. Second, we then
develop a sequential hypothesis test based on
a sequential fixed width one-sided confidence interval (CI) and show
how it can overcome some of the limitations of 
traditional SPRT. We empirically find that even this test does not
always perform better than the more traditional SPRT.
Third, building on the above, we propose a hybrid approach (HYB), where we dynamically
select the hypothesis test (SPRT or CI) based on the similarity of each
candidate pair which we crudely estimate from the first few
comparisons. In other words, {\it instead of using a single
fixed hypothesis test, we select one which is best suited for the
candidate pair being estimated.}

For a candidate pair $x,y$ with similarity $s$, the
probability of a hash matching is $s$ for a locality
sensitive hash function as described in equation~\ref{eq:lsh}.
Therefore, given $n$ hashes for the pair, the probability
that $m$ of them will match follows a binomial distribution
with parameters $n$ and $s$. This is because the individual
hash matching probabilities are identically and independently distributed 
Bernoulli with parameter $s$. So our problem formulation
reduces to doing sequential hypothesis test on a binomial
parameter $s$. 

\subsubsection{Sequential Probability Ratio Test}
We use the traditional sequential probability ratio test
by Wald~\cite{wald1973sequential} as our first principled sequential model
for matching LSH signatures, to decide between $s\geq t$ or $s<t$. For the
purpose of this model we swap the null and alternate
hypotheses of our formulation. We do this because the
resulting formulation of the hypothesis test $H0:s<t$
vs. $H1:s\geq t$, where $s$ is a binomial parameter, has a well known 
textbook solution (due to Wald). 
The important thing
to recollect is that we care more about Type I error
in our original formulation.  Therefore under the swapped SPRT setting,
we care about the Type II error. That is
easily done as SPRT allows the user to set both Type I
and Type II errors, and we set Type II error to be
$\alpha$. To solve a composite hypothesis test using SPRT
for a binomial parameter $s$, the first step is to
choose two points $t+\tau$ and $t-\tau$. Now
the SPRT becomes a simple hypothesis test problem
of $H0:s=s_0=t-\tau$ vs. $H1:s=s_1=t+\tau$. The algorithm
 works as follows:
\begin{enumerate}
\item Incrementally compare batches of size $b$
hashes until
\begin{equation*}
\begin{split}
&\frac  {log(\frac{\alpha}{1-\beta})}{log(\frac{s_1}{s_0})-log(\frac{1-s_1}{1-s_0})}
       +n \frac{log(\frac{1-s_0}{1-s_1})}{log(\frac{s_1}{s_0})-log(\frac{1-s_1}{1-s_0})}
			< \hat{s} \\
			&\ \ \ \ < \frac{log(\frac{1-\alpha}{\beta})}{log(\frac{s_1}{s_0})-log(\frac{1-s_1}{1-s_0})}
			+n \frac{log(\frac{1-s_0}{1-s_1})}{log(\frac{s_1}{s_0})-log(\frac{1-s_1}{1-s_0})}
\end{split}
\end{equation*}
Here $n$ is the cumulative number of hash comparisons till now,
and $\hat{s} = m/n$, where $m$ is the cumulative number of hash matches up to that point.
\item Reject null hypothesis (conclude $s\geq t$) if,
\begin{equation*}
%\begin{split}
\hat{s} \geq \frac{log(\frac{1-\alpha}{\beta})}{log(\frac{s_1}{s_0})-log(\frac{1-s_1}{1-s_0})}
			+n \frac{log(\frac{1-s_0}{1-s_1})}{log(\frac{s_1}{s_0})-log(\frac{1-s_1}{1-s_0})}
%\end{split}
\end{equation*}
\item Fail to reject null hypothesis (conclude $s<t$) if,
\begin{equation*}
%\begin{split}
\hat{s} \leq \frac  {log(\frac{\alpha}{1-\beta})}{log(\frac{s_1}{s_0})-log(\frac{1-s_1}{1-s_0})}
       +n \frac{log(\frac{1-s_0}{1-s_1})}{log(\frac{s_1}{s_0})-log(\frac{1-s_1}{1-s_0})}
			%\end{split}
\end{equation*}
\end{enumerate}

SPRT is a cumulative likelihood ratio test, and is an optimal
test with guaranteed Type I and Type II errors, when the
hyotheses are simple. In the case of composite hypothesis (across bins of hashes),
no optimality guarantees can be given, and consequently,
to make a decision, SPRT typically takes a large number
of hash comparisons. This results in extremely slow performance
as we will empirically validate. We next describe the confidence interval based test.
 
\subsubsection{One-Sided-CI Sequential Hypothesis Test}
\label{sec:one-ci-sec}
\paragraph{Constructing the confidence interval (CI)}
The true similarity $s$ of pair of data objects $x,y$ can
be estimated as $\hat{s} = \frac{m}{n}$, where $m$ is the
number of hashes that matched out of $n$ hash comparisons.
It can be shown that $\hat{s}$ is the maximum likelihood
estimate of $s$~\cite{rice2007mathematical}. Following standard convention
we denote this estimator as $\hat{S}$ (random variable, distinguished
from its realization, $\hat{s}$).
Here we describe the procedure for constructing a fixed-width
(say $w$) upper confidence interval for $s$ with $1-\alpha$
 coverage probability. More formally, we want to continue
comparing hashes and estimating similarity until,
\begin{equation}
P(s < \hat{S} + w) = 1- \alpha
\label{eq:upper-limit}
\end{equation}
Here $\hat{s} + w$
is the upper confidence limit for $s$ with $1-\alpha$ coverage probability.
We use an approach similar to Frey~\cite{frey2010fixed} to solve this
problem.

\smallskip
\noindent \textbf{Stopping rule for incremental hash comparisons:}
We use the Wald confidence interval for binomial as our
stopping rule. Formally, for some value $\lambda$, and
a fixed confidence width $w$, we incrementally
compare batches of $b$ hashes and stop when $z_{\lambda} \sqrt\frac{\hat{s_a}(1-\hat{s_a})}{n} \leq w$.
Then the upper confidence limit can be reported as $min(\hat{s}+w, 1.0)$. Here $\hat{s_a} = \frac{m+a}{n+2a}$,
where $a$ is a very small number.  $\hat{s_a}$ is used instead of $\hat{s}$ because, if the batch size is extremely small,
and number of matches is 0 (or it is the maximum, i.e., all match), then the confidence width
becomes 0 after the first batch if $\hat{s}$ is used. 

\smallskip
\noindent \textbf{Finding $\lambda$:} In a
non-sequential setting, the Wald upper confidence
limit as described above will have a coverage
probability of $1-\lambda$. But in
a sequential setting, where the confidence interval
is tested after every batch of hash comparisons,
the coverage probability could fall below $1-\lambda$.
Hence to ensure coverage probability of at least
$1-\alpha$, $\lambda$ should be set less than $\alpha$.
Given the set of stopping points and a $\lambda$,
we can compute the coverage probability CP($\lambda$)
of our one-sided confidence interval using the
path-counting technique~\cite{girshick2006unbiased}. Suppose the stopping points
are $(m_1,n_1),(m_2,n_2),.....,(m_k,n_k)$ and
$H(m,n)$ is the number of ways to do $n$ hash
comparisons with $m$ matches, without hitting
any of the stopping points. Therefore, the probability
of stopping at stopping point $(m_i,n_i)$ is
$H(m_i,n_i)s^{m_i}(1-s)^{n_i-m_i}$. Since the
incremental hash comparison process is guaranteed
to stop, probability of stopping at all the
stopping points should sum to 1. This implies
\begin{eqnarray*}
\sum\limits_{i=1}^k H(m_i,n_i)s^{m_i}(1-s)^{n_i-m_i} = 1
\end{eqnarray*}
Consequently, the
coverage probability for similarity $s$ will be
\begin{equation}
T(s,\lambda) = \sum\limits_{i=1}^k H(m_i,n_i)s^{m_i}(1-s)^{n_i-m_i} I(s \leq \frac{m_i}{n_i}+w)\\
\label{eq:cov-prob}
\end{equation}
Here $I$ is the indicator function. Now the coverage probability can be computed as,
\begin{equation}
CP(\lambda) = min_{s \in [0,1]} T(s,\lambda)
\label{eq:confidence-coef}
\end{equation}
For our one-sided confidence interval to have
at least $1-\alpha$ coverage probability, we
need to find $\lambda$ such that $CP(\lambda) \geq 1-\alpha$.
The function $H(m,n)$ can be solved using the path-counting
recurrence relation:
\begin{equation*}
\begin{split}
H(m,n+1) = &H(m,n)  ST(m,n) \\
&+ H(m-1,n) ST(m-1,n)
\end{split}
\end{equation*}
Here $ST(m,n)$ is the indicator function of whether
$(m,n)$ is a stopping point. The base of the
recursion is $H(0,1) = H(1,1) = 1$.
With a fixed $\lambda$, $ST(m,n)$ can be computed
using the Wald stopping rule as described before,
and $H(m,n)$ can be hence computed by the aforementioned recurrence
relation. Then we need to solve equation~\ref{eq:confidence-coef}
to find out the confidence coefficient of our one-sided
interval. 
$T(s,\lambda)$ is a piecewise polynomial in $s$ with
jumps at the points in the set $C=\{0,\frac{m_i}{n_i}+w, \forall i=1\ to\ k\ and\ \frac{m_i}{n_i}+w \leq 1\}$.
$CP(\lambda)$ is then approximated numerically by
setting $s=c\pm 10^{-10},$ where $c\in C$ and taking
the minimum resulting $T(s,\lambda)$.
Now that we know how to compute $CP(\lambda)$, we use bisection root-finding
algorithm to find a $\lambda$ for which $CP(\lambda)$
is closest to our desired coverage probability $1-\alpha$.

\paragraph{Constructing the Hypothesis Test}
In the previous section we described a procedure
for creating a fixed-width once-sided sequential
upper confidence limit with coverage probability
$1-\alpha$. In this section, we describe
the process to convert the one-sided upper
confidence interval to a level-$\alpha$
hypothesis test using the duality of confidence
intervals and hypothesis tests.
\begin{lemma}
\label{level-alpha}
If $\hat{s}+w$ be an upper confidence limit for $s$
with coverage probability $1-\alpha$, then a $level-\alpha$
hypothesis test for null hypothesis $H0:s\geq t$ against
alternate hypothesis $H1:s < t$ will be
Reject H0, if $\hat{s}+w < t$,
else Fail to Reject H0. 
\end{lemma}
\begin{proof}
By equation~\ref{eq:upper-limit},
\begin{flalign*}
&P(\hat{S}+w \geq s) = 1-\alpha&\\
&\implies P(\hat{S}+w \geq t | s \geq t) \geq 1-\alpha\\
&\implies -P(\hat{S}+w \geq t | s \geq t) \leq -1 + \alpha\\
&\implies 1 - P(\hat{S}+w \geq t | s \geq t) \leq \alpha\\
&\implies P(\hat{S}+w < t | s \geq t) \leq \alpha\\
&\implies P (Reject H0 | H0) \leq \alpha \\
\end{flalign*}
\end{proof}

\paragraph{Choosing $w$}
The fixed-width $w$ of the one-sided upper confidence
interval has a significant effect on the efficiency of the
test. Intuitively, the larger the width $w$, the less
the number of hash comparisons required to attain
a confidence interval of length $w$. However, setting
$w$ to a very high value would result in a large Type II
error for our test. Though our algorithm's quality
is not affected by Type II error (since we compute
exact or approximate similarity when alternate hypothesis is satisfied),
but still a large Type II error will imply that
many false positives (candidates which fall in alternate
hypothesis, but are classfied as null hypothesis).
Exact similarity is computed or approximate
similarity is estimated and these candidates are
pruned away. Therefore, a large Type II error will
translate to lower efficiency. In other words,
making $w$ too high or too low will cause significant
slowdown in terms of execution time.

We next describe a simple heuristic to select $w$.
Suppose a candidate pair has similarity $s < t$,
i.e. for this candidate pair, the null hypothesis
should be rejected. So the upper confidence limit
$\hat{s}+w$ can be as high as $t$, and our
test statistic should still be able to reject it. Therefore, the maximum
length of 
$w$ is dictated by how large the upper confidence
limit can be. So instead of preseting $w$ to a fixed
value, we dynamically set $w$ according to the following
heuristic. We compare the first batch of hashes and
use the crude estimate of $s$, say $\hat{si}$ from the
first batch to come up with $w$:
\begin{equation}
w = t-\hat{si} - \epsilon
\end{equation}
The key insight here is, \textit{instead of using
a single hypothesis test for all candidate pairs,
we choose a different hypothesis test based on
an initial crude similarity estimate of the candidate
pair being analyzed, so that $w$ can be maximized,
while still keeping Type II error in control, resulting
in efficient pruning.} Note that every such test is
a $level-\alpha$ test according to Lemma~\ref{level-alpha}.
We need the $\epsilon$ parameter as $\hat{si}$ is a 
crude estimate from the first batch of hash comparisons
and it could be an underestimate, which would result in
an overestimate of $w$. Consequently, the final
test statistic $\hat{s}+w$ can go beyond $t$ and the
candidate cannot be pruned. Figure~\ref{fig:crude-estimate}
explains the phenomenon. Ofcourse, dynamically
constructing the test for each candidate can make
the candidate pruning process inefficient. We solve this issue
by caching a number of tests for different $w$ and during
the candidate pruning step, the test that is closest to $w$ (but smaller
than or equal to it)
is selected. Hence there is no need for online inference,
making the algorithm very efficient.

\begin{figure}
\begin{center}
\includegraphics[width=200pt,height=100pt]{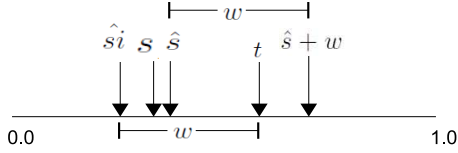}
\end{center}
\caption{Estimating w}
\label{fig:crude-estimate}
\end{figure}

\subsubsection{Hybrid Hypothesis Tests}
We found out empirically that the number of hash comparisons
required by SPRT is very high for our composite hypothesis
test problem. The one-sided CI based tests performed
considerably better. Specifically we saw that the candidate
pairs whose actual similarity $s$ is quite far away from the
threshold $t$ were very quickly classified into the null or
alternate hypothesis and the width $w$ is quite large.
But interestingly, the candidate pairs which have similarity
very close to the threshold, the estimated parameter $w$
becomes very small. For such pairs, to attain the fixed-width
confidence interval, the number of hash comparison requirement
is very high. It is even higher than the more traditional
SPRT. Therefore, to utilize the best of both worlds, we use
a hybrid hypothesis testing strategy, where based on how far the
true similarity is away from the threshold, we either select
a one-sided CI based hypothesis test, or the SPRT. Formally,
\textit{we use a parameter $\mu$, such that if
the estimated fixed-width $w \geq \mu$, we use the
one-sided CI based hypothesis test, else we use SPRT.}. Again,
in this hybrid strategy all the tests are $level-\alpha$, so
we have guarantees on the overall Type I error, while
minimizing the number of overall hash comparisons by smarty
selecting the proper test for a specific candidate.

\subsection{Concentration Inference}
\label{sec:two-sided-ci}
To solve the concentration inference of equation~\ref{eqn:concentrateProb},
we can 
create a two-sided fixed-width confidence interval for a 
binomial proportion under the sequential setting. The technique is
very similar to the one we described in section~\ref{sec:one-ci-sec}
for the one-sided upper confidence limit. The major difference
is, for two-sided confidence interval, the coverage probability
equation~\ref{eq:cov-prob} from section~\ref{sec:one-ci-sec} becomes:
\begin{equation*}
T(s,\lambda) = \sum\limits_{i=1}^k H(m_i,n_i)s^{m_i}(1-s)^{n_i-m_i} I(|s - \frac{m_i}{n_i}| \leq \delta)\\
\end{equation*}
Now this equation can be solved in a manner similar to the one
described in section~\ref{sec:one-ci-sec} to find out the
critical value $\lambda$ and hence the stopping points
in the sequential process. The stopping rule will also
change to $z_{\frac{\lambda}{2}} \sqrt\frac{\hat{s_a}(1-\hat{s_a})}{n} \leq \delta$
(in the one-sided case $\lambda$ was used instead of $\frac{\lambda}{2}$).
The concentration inference is
used to estimate the similarity with probabilistic
guarantees under circumstances where exact similarity
computation is infeasible.

\noindent \textbf{Choosing maximum number of hashes:}
In our problem scenario, we do not need all candidates
to converge to a fixed-width interval. Since we guarantee $1-\alpha$
recall, any candidate pair which has less than $\alpha$ probability
of being greater than $t$ can be ignored. In other words,
we do not need to consider all the stopping
points generated. We can choose the stopping points based
on the user defined threshold $t$.
\begin{lemma}
\label{stopping-points}
If $m_i,n_i$  are the stopping points decided by the fixed-width
confidence interval method having coverage probability $\gamma$, 
only the stopping points $m_i,n_i$,
such that $\frac{m_i}{n_i} < t - \delta$ will have probability
$\gamma$ of being greater than threshold $t$.
\end{lemma}
\begin{proof}
By the fixed-width confidence interval guarantee, 
\begin{flalign*}
&P(\hat{S}-\delta \leq s \leq \hat{S}+\delta ) = 1 - \gamma& \\
&\implies P(s\geq \hat{S}+\delta\ or\ s\leq\hat{S}-\delta) = \gamma\\
&\implies P(s \geq \hat{S}+\delta) \leq \gamma
\end{flalign*}
This implies for any $\hat{s} < t - \delta$, $s>t$ will have coverage
probability less than $\gamma$. For all the stopping points $m_i,n_i$,
$\hat{s} = \frac{m_i}{n_i}$. Hence the proof.
\end{proof}
\begin{corollary}
If $m_i,n_i$ are the set of stopping points, the maximum number of
hashes $n_{max}$ required by our algorithm to estimate any similarity
above $t$ is $n_{max} = max(n_i)$ s.t. $\frac{m_i}{n_i} \geq t - \delta$.
\end{corollary}
If we set $\gamma = \alpha$, the above lemma will ignore points
which have less than $\alpha$ probability of being a true positive.
In otherwords, our algorithm is able to guarantee $1-\alpha$ recall.

\subsection{Similarity Measures}
The proposed techniques in this paper can be used
for the set similarity measures for which a locality
sensitive hash function exists. Previous work has
developed locality sensitive hash functions for a
wide range of similarity measures~\cite{indyk98,datar04,broder97,ravichandran05,henzinger06}
as well as for arbitrary kernel functions~\cite{kulis2012kernelized}.
In this paper, we show that our methods work well with
two of the most popular similarity measures - i) Jaccard similarity
and ii) Cosine similarity.

\subsubsection{Jaccard Similarity}
The
locality sensitive hash function relevant to Jaccard similarity is \textit{MinWise Independent
Permutation}, developed by Broder \textit{et al}~\cite{broder97}. This hash function
can approximate the Jaccard coefficient between two sets $x,y$. Formally,
\begin{equation*}
P(h(x) == h(y)) = \frac{|x\cap y|}{|x\cup y|}
\end{equation*}
The estimate of Jaccard similarity between $x,y$
will be:
\begin{equation*}
\hat{s} = \frac{1}{n} \sum_{i=1}^n I[h_i(x) == h_i(y))]
\end{equation*}
As described earlier in equation~\ref{eq:likelihood}, the likelihood function for
getting $m$ matches out $n$ hashes is a binomial with parameters $n,s$.
Note that $n$ is also a random variable here. Hence we can directly use
our proposed methods for doing inference on $s$.

\subsubsection{Cosine Similarity}
The
locality sensitive hash function for cosine similarity is given by
the \textit{rounding hyperplane algorithm}, developed by Charikar~\cite{charikar02}.
However, the similarity given by the above algorithm is a little different
from cosine similarity. Specifically, such a hash function gives:
\begin{equation*}
P(h(x) == h(y)) = 1 - \frac{\theta}{\pi}
\end{equation*}
where $\theta$ is the angle between the two vectors. Let the above
similarity be defined as $s$ and let the cosine similarity be
$r$. The range of $s$ is therefore, 0.5 to 1.0.
To convert between $s$ and $r$, we need the following transformations:
\begin{eqnarray}
r =  cos(\pi(1-s)) \\
s = 1 - \frac{cos^{-1}(r)}{\pi}
\end{eqnarray}
Consequently, we need to adapt our proposed algorithms to handle
these transformations. Handling the pruning inference is quite
simple. If the user sets the cosine similarity threshold
as $t$, before running our pruning inference, we change the
threshold to the value of the transformed similarity measure. So
the pruning inference becomes $s \geq (1 - \frac{cos^{-1}(t)}{\pi})$
instead of $s \geq t$. 

The transformation of the concentration inference
is trickier. We need to transform the confidence interval of $s$ (
our algorithm will generate this) to the confidence interval of $r$
(for cosine similarity). The user provides an estimation error
bound $\delta$, implying that we need
to generate an estimate $\hat{r}$ within a confidence interval of $2\delta$.
with $1-\gamma$ coverage probability. Since we can only estimate $\hat{s}$,
we need to create a level-$(1-\gamma)$ confidence interval $2\delta_s$ around $\hat{s}$,
such that, if $l_s \leq s \leq u_s$ and $l_r \leq r \leq u_r$
then,
\[
u_s-l_s \leq 2\delta_s \implies u_r - l_r \leq 2\delta
\]
If we create, a $2\delta_s$ fixed-width confidence interval,
then the upper and lower confidence limits will be $\hat{s}+\delta_s$
and $\hat{s}-\delta_s$ respectively. Since $r$ is a monotonicall increasing
function of $s$, hence the upper and lower confidence limit of $r$ (cosine
similarity) will be $cos(\pi(1-min(1.0,\hat{s}+\delta_s)))$ and
$cos(\pi(1-max(0.5,\hat{s}-\delta_s)))$ respectively. The interval for the estimate
of cosine similarity will be $cos(\pi(1-min(1.0,\hat{s}+\delta_s))) - cos(\pi(1-max(0.5,\hat{s}-\delta_s)))$.
Now we have to choose $\delta_s$ such that 
\[
cos(\pi(1-min(1.0,\hat{s}+\delta_s))) - cos(\pi(1-max(0.5,\hat{s}-\delta_s))) \leq 2\delta
\]
The above function is monotonically decreasing in $\hat{s}$. So the interval
will be largest when $\hat{s}$ is smallest (0.5). So we set $\hat{s}$ to 0.5
and numerically find out largest $\delta_s$ such that the inequality
of the above expression is still satisfied.

\section{Experimental Evaluation}
\subsection{Experimental Setup and Datasets}
In terms of setup all results are run on 
a single processor on 
a AMD Opteron 2378 with 2.4GHz cpu speed. The machine
has 32GB of RAM. We only use one core as our application
is a single threaded C++ program. We use
two real world dataset to evaluate the quality and
efficiency of our allpairs similarity search algorithm.
Details are given in Table~\ref{tab:datasets}

\noindent \textbf{Twitter:} This is a graph representing
follower-followee links in Twitter~\cite{kwak10}. Only users
having at least 1000 followers are selected. Each user
is represented as an adjacency list of the users it follows.

\noindent \textbf{WikiWords100K and WikiLinks:} These datasets
are derived from the English Wikipedia Sep 2010 version.
The WikiWords100K is a preprocessed text corpus with
each article containing at least 500 words. The Wikilinks
is a graph of created from the hyperlinks of the entire
set of articles weighted by tf-idf.

\noindent \textbf{RCV:} This is a text dataset consisting of
reuters articles~\cite{lewis04}. 
Each user is represented as a set of words. Some basic
preprocessing such as stop-words removal and stemming
is done.

\noindent \textbf{Orkut:} The Orkut dataset is a friendship
graph of 3 million users weighted by tf-idf~\cite{mislove07}.

\begin{table}
\begin{small}
\begin{tabular}{|c|c|c|c|c|}
\hline
{\bf
Dataset} & {\bf Vectors} & {\bf Dimensions} & {\bf Avg. Len} &
{\bf Nnz}\\
\hline
Twitter & 146,170 & 146,170 & 1369 & 200e6\\
\hline
WikiWords100K & 100,528 & 344,352 & 786 & 79e6\\
\hline
RCV & 804,414 & 47,236 & 76 & 61e6\\
\hline
WikiLinks & 1,815,914 & 1,815,914 & 24 & 44e6\\
\hline
Orkut & 3,072,626 & 3,072,66 & 76 & 233e6\\
\hline

\end{tabular}
\end{small}
\caption{Dataset details}
\label{tab:datasets}
\end{table}

\subsection{Results}
As explained in the methodology section~\ref{sec:mlest},
we expect BayesLSH/Lite variants to be very fast, however those could
potentially suffer a loss in the qualitative guarantees
as they model an inherently sequential process in a non-sequential
manner. Since the sequential confidence interval based methods have
provable guarantees about quality (lemmas~\ref{level-alpha} and~\ref{stopping-points}),
they are always expected to be qualitatively better. The SPRT and hence the
Hybrid methods should qualitatively perform very well, however under
the composite hypothesis testing scenario, SPRT cannot provide strong
guarantees. In the next two sections we will evaluate these premises.

\subsubsection{Algorithms using Early Pruning and Exact Similarity Computation}
We compare the following four strategies for computing
all pairs with similarity above a certain user defined
threshold. All of these algorithms assume that the original
dataset is available (instead of the smaller sketch). These
algorithms use the exact candidate generation
technique AllPairs~\cite{bayardo07} and an early pruning technique
and finally exact similarity computation.

\noindent \textbf{BayesLSHLite:} BayesLSHLite~\cite{satuluri2012bayesian}
is the state-of-the-art candidate pruning algorithm which is known
to perform better than AllPairs~\cite{bayardo07} and PPJoin~\cite{xiao08}.

\noindent \textbf{SPRT:} We use the traditional Sequential Probability
Ratio Test to do the early pruning of candidates. We set $\tau = 0.025$.

\noindent \textbf{One-Sided-CI-HT:} We compare against our model,
which is the fixed-width one-sided upper confidence
interval based hypothesis testing technique. We set $\epsilon=0.01$.
The choice of $\epsilon$ is done by empirically evaluating several values -- we found
values in the neighborhood of $0.01-0.05$ worked best.
We set $a=4$ as it seems to work well in practice~\cite{frey2010fixed}.

\noindent \textbf{Hybrid-HT:} This is the second model we propose,
where based on the candidate in question, we either choose
a One-Sided-CI-HT or SPRT. We set $\mu = 0.18$, that is
the threshold of $w$ below which, our Hybrid-HT algorithm
switches from a One-Sided-CI-HT to SPRT. Again, we selected $\gamma$,
empirically by trying different thresholds.
For all the tests above, we set the Type I error or recall parameter
$1-\alpha = 0.97$. We also compare with the AllPairs algorithm
which uses exact similarity computation right after the candidate
generation step (no candidate pruning).

The performance and quality numbers  are reported
in Figure~\ref{fig:perf-qual}. We measure performance by the
total execution time and we measure quality by recall (since we
are giving probabilistic guarantees about recall).
An added advantage of these methods is
that since we compute exact similarity for all candidates that are retained
and
check whether they are above the threshold using exact similarity computation, all the strategies
yield full precision (100\%). 
Further more, the sequential hypothesis
tests we do are truncated tests, i.e. we compute at most $h=256$
hashes, after which if a decision cannot be made, we send the pair
for exact similarity computation. We report results on all the
aforementioned datasets on both Jaccard and cosine similarity
measures. For Jaccard, we vary the similarity threshold from $0.3-0.7$
and for cosine, we vary the threshold from $0.5-0.9$. These are the
same parametric settings used in the original BayesLSHLite work.

Results indicate that the pattern is quite similar
for all the datasets. BayesLSHLite is always substantially faster in case
of cosine similarity while in case of Jaccard similarity, AllPairs
is marginally faster at times. At high values of
the similarity threshold, SPRT is the slowest, while both One-Sided-CI-HT
and Hybrid-HT performs very close to BayesLSHLite. This performance
benefit comes from the one-sided tests. More precisely, choosing the width $w$
of the test based on the estimate first bin of hash comparisons makes each
test optimized for the specific candidate being processed.
Those tests are extremely efficient
at pruning away false positive candidates whose true similarities are very far
from the similarity threshold $t$. The reason is these tests can allow a larger
confidence width $w$ and hence less number of hash comparisons. Even the Hybrid-HT
performs very well at such high thresholds, because it chooses one of the
one-sided tests at such high threshold. However, at the other end of the
spectrum, at very low similarity thresholds, the allowable confidence interval
width $w$ becomes too small and a large number of trials is required by the
one-sided tests making them inefficient. SPRT performs reasonably well under these
situations. Under these conditions, the Hybrid-HT strategy is able to perform better
than both (SPRT and One-Sided-CI-HT) as it able to smarty delegate pairs with
true similarity close to threshold to
SPRT instead of one-sided tests. 
In summary, the green lines (Hybrid-HT) can perform well through the
whole similarity threshold range. For the WikiWords100K dataset in 
Figure~\ref{fig:wiki100-cos-256} Hybrid-HT gave 8.8x speedup over
AllPairs and 2.1x speedup of SPRT at 0.9 threshold and at 0.5
threshold, it gave 3.4x speedup over AllPairs and 1.3x speedup
over SPRT.

In terms of quality, our proposed method One-Sided-CI-HT guarantee at least 97\%
recall ($\alpha = 0.03$). In all results we see the recall of One-Sided-CI-HT,
as expected,
is above 97\%. Inspite of the fact that SPRT does not have strong guarantees in case of composite hypothesis,
 we see that SPRT performs quite well in all datasets.
Since Hybrid-HT uses the One-Sided-CI-HT and SPRT, its quality numbers are also
extremely good.
Only BayesLSHLite, which does not model the hash comparisons
as a sequential process, falls marginally below the 97\% mark at some places. 
In summary, our tests can provide rigorous
quality guarantees, while significantly improving the performance by
over traditional SPRT.

\begin{figure*}[!htb]
\begin{center}
\includegraphics[width=280pt, height=20pt]{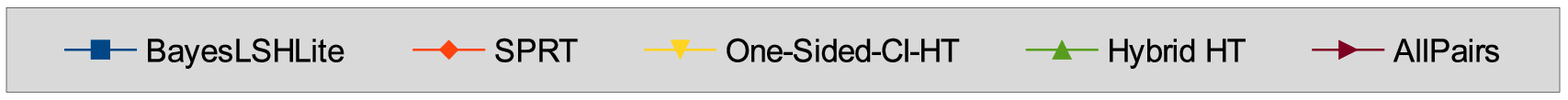}
\end{center}
%\vspace{-0.33in}
\begin{center}
\hspace{-0.5in}
\subfigure[Twitter, Jaccard]
{
\includegraphics[width=120pt,height=100pt]{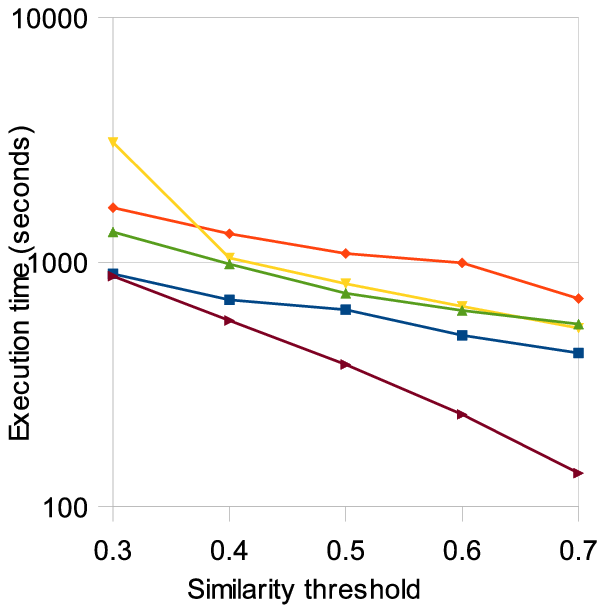}
\label{fig:twit-jac-256}
}
\subfigure[Twitter, Jaccard]
{
\includegraphics[width=120pt,height=100pt]{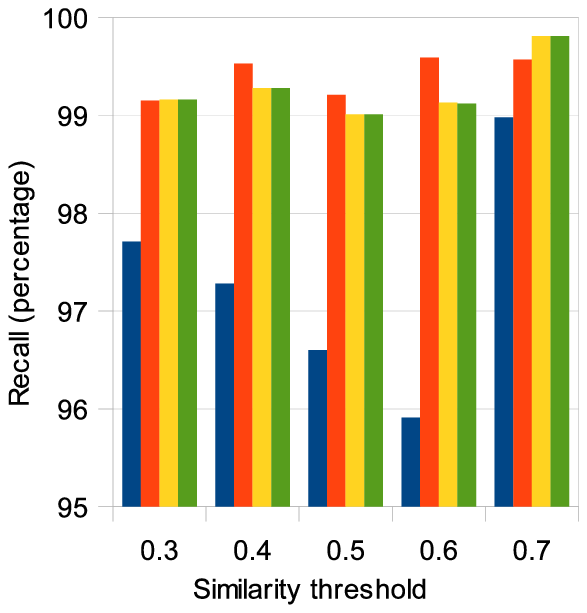}
\label{fig:twitter-jac-256-quality}
}
\subfigure[WikiWords100K,Jaccard]
{
\includegraphics[width=120pt,height=100pt]{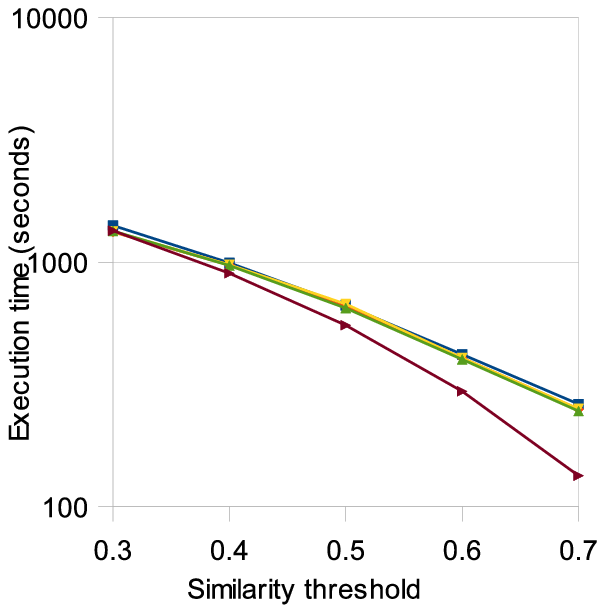}
\label{fig:wiki100-jac-256}
}
\subfigure[WikiWords100K,Jaccard]
{
\includegraphics[width=120pt,height=100pt]{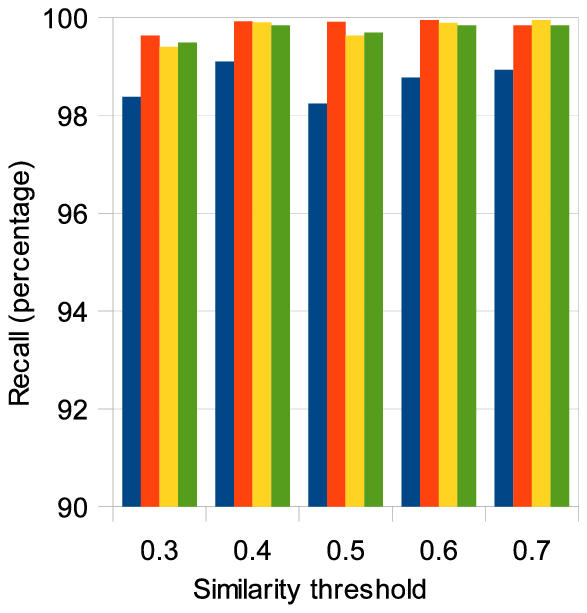}
\label{fig:wiki100-jac-256-quality}
}
\end{center}
\begin{center}
\hspace{-0.5in}
\subfigure[RCV,Jaccard]
{
\includegraphics[width=120pt,height=100pt]{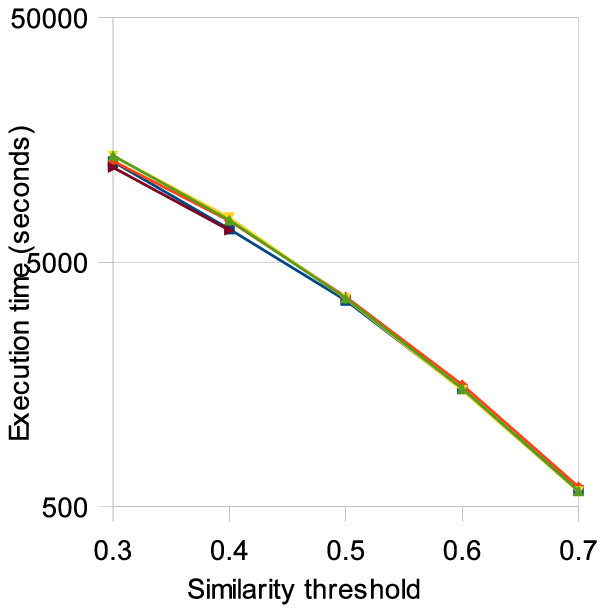}
\label{fig:rcv-jac-256}
}
\subfigure[RCV,Jaccard]
{
\includegraphics[width=120pt,height=100pt]{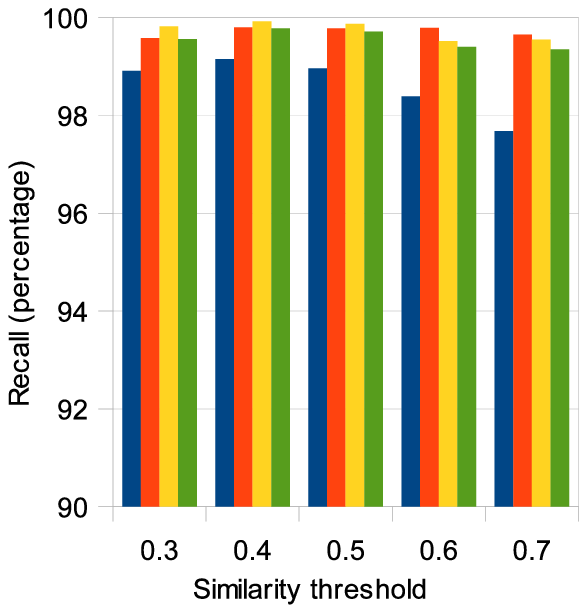}
\label{fig:rcv-jac-256-quality}
}
\subfigure[WikiLinks,Jaccard]
{
\includegraphics[width=120pt,height=100pt]{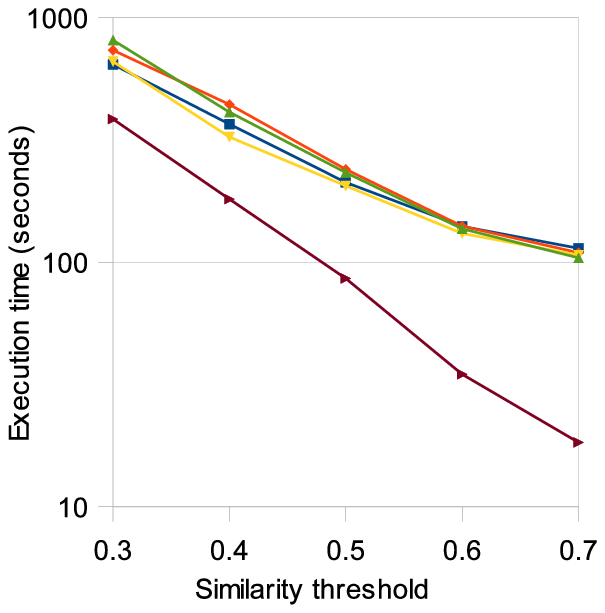}
\label{fig:wikilinks-jac-256}
}
\subfigure[WikiLinks,Jaccard]
{
\includegraphics[width=120pt,height=100pt]{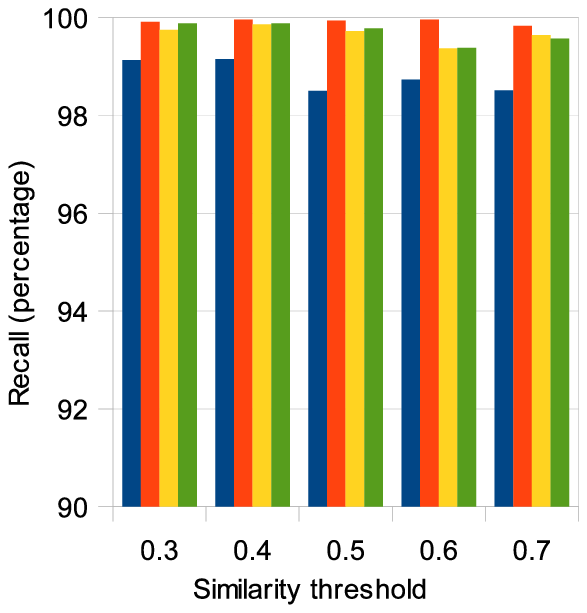}
\label{fig:wikilinks-jac-256-quality}
}
\end{center}

\begin{center}
\hspace{-0.5in}
\subfigure[Twitter,cosine]
{
\includegraphics[width=120pt,height=100pt]{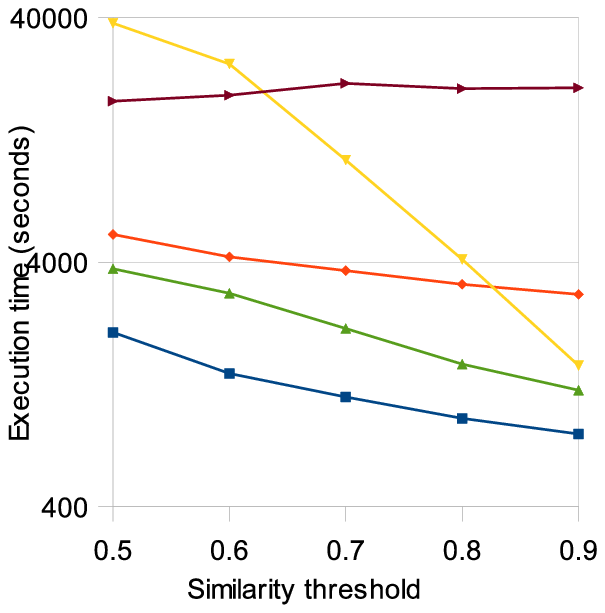}
\label{fig:twit-cos-256}
}
\subfigure[Twitter,cosine]
{
\includegraphics[width=120pt,height=100pt]{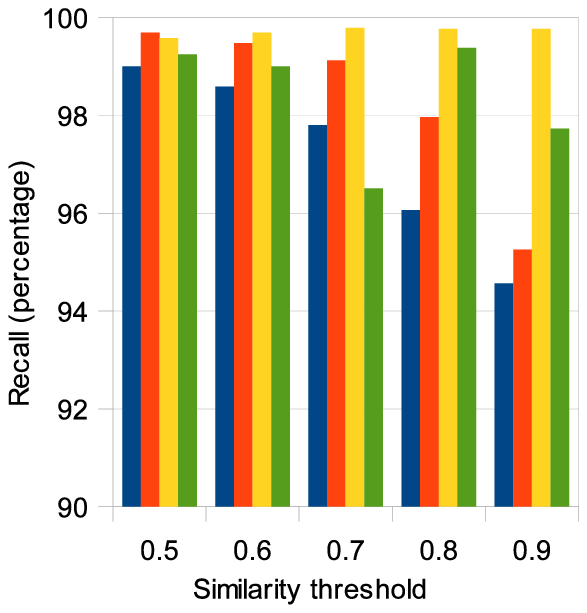}
\label{fig:twitter-cos-256-quality}
}
\subfigure[WikiWords100K,cosine]
{
\includegraphics[width=120pt,height=100pt]{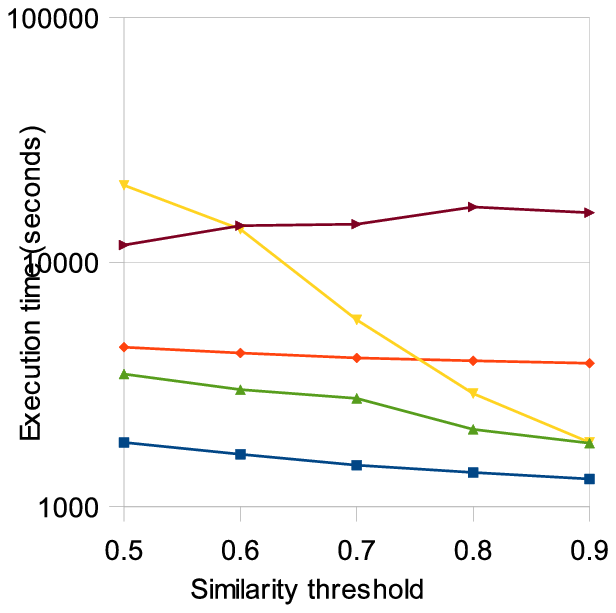}
\label{fig:wiki100-cos-256}
}
\subfigure[WikiWords100K,cosine]
{
\includegraphics[width=120pt,height=100pt]{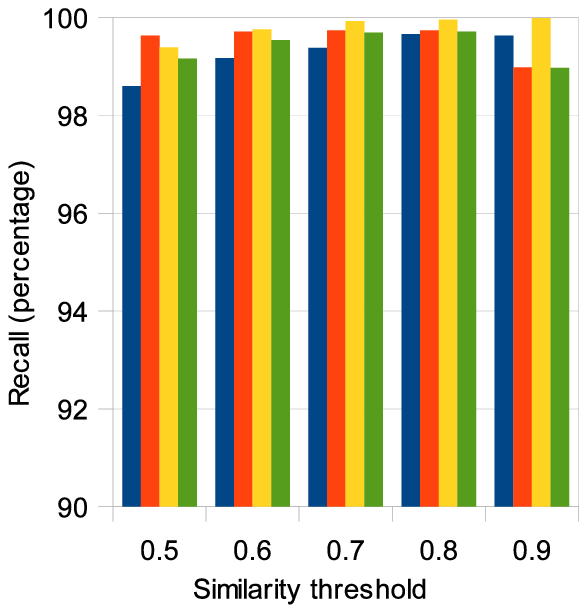}
\label{fig:wiki100-cos-256-quality}
}
\end{center}

\begin{center}
\hspace{-0.5in}
\subfigure[RCV,cosine]
{
\includegraphics[width=120pt,height=100pt]{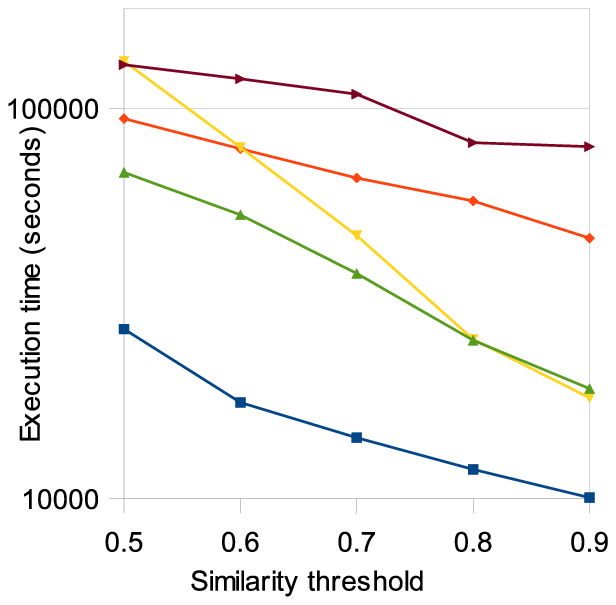}
\label{fig:rcv-cos-256}
}
\subfigure[RCV,cosine]
{
\includegraphics[width=120pt,height=100pt]{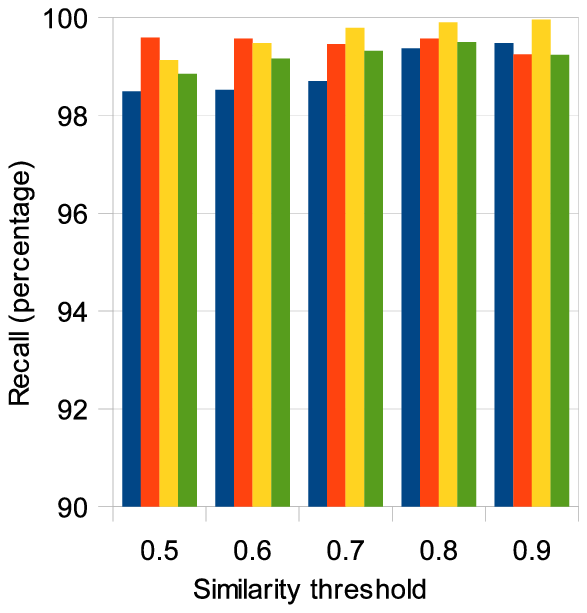}
\label{fig:rcv-cos-256-quality}
}
\subfigure[WikiLinks,cosine]
{
\includegraphics[width=120pt,height=100pt]{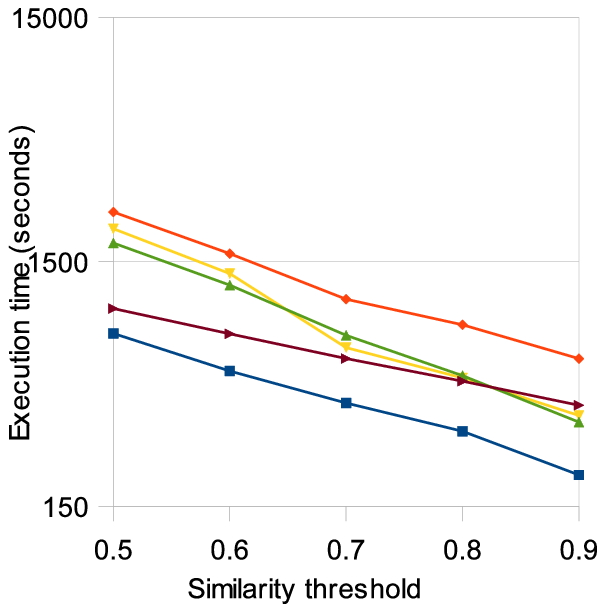}
\label{fig:wikilinks-cos-256}
}
\subfigure[WikiLinks,cosine]
{
\includegraphics[width=120pt,height=100pt]{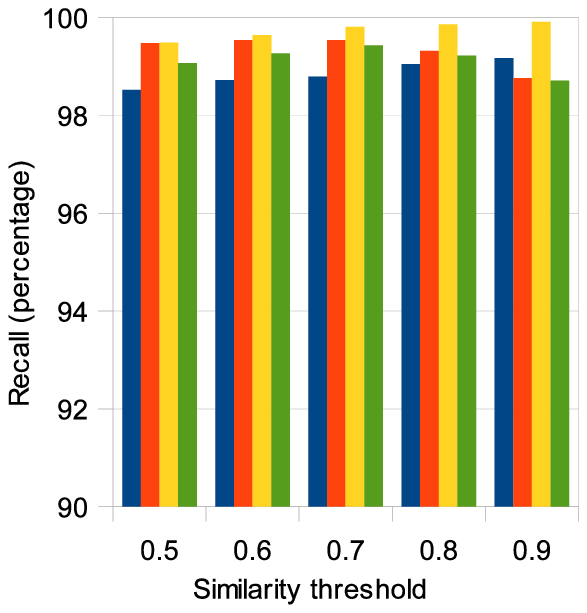}
\label{fig:wikilinks-cos-256-quality}
}
\end{center}

\begin{small}
\caption{Comparisons of algorithms with exact similarity computation.  
}
\label{fig:perf-qual}
\end{small}
\end{figure*}

\subsubsection{Algorithms using Early Pruning and Approximate Similarity Estimation}
The previous section discussed the algorithms which
can be used when the explicit representation of the
original data is available. We now describe results
on two algorithms for which only the hash signatures
needs to be stored rather than the entire dataset.
These algorithms use the LSH index generation followed
by candidate pruning, followed by approximate similarity
estimation. We compare the following two techniques:

\noindent \textbf{BayesLSH:} This uses the same
pruning technique as BayesLSHLite along with
the concentration inference for similarity
estimation.

\noindent \textbf{Hybrid-HT-Approx:} This is our
sequential variant. It uses Hybrid-HT's pruning
technique along with the sequential fixed-width
confidence interval generation strategy as described
in section~\ref{sec:two-sided-ci}. We set $\tau = 0.015$.

We use the same parametric settings as before. The additional
parameters required here are the estimation error bound
$\delta$ and the coverage probability for the confidence
interval $\gamma$. We set $\delta=0.05$ and $\gamma=\alpha$. Again we measure performance
by execution time. Here we measure quality by both recall
and estimation error as we provide probabilistic
guarantees on both.

Figure~\ref{fig:perf-qual-lsh} reports both the performance
and recall numbers. We do not list the estimation error
numbers as the avg. estimation error for each algorithm
on each dataset was within the specified bound of 0.05.
Results indicate that Hybrid-HT-Approx is slower than
BayesLSH as expected, however is qualitatively better
than BayesLSH. More importantly, in all cases,
Hybrid-HT-Approx has a recall value which is well above
the 97\% guaranteed number. 
BayesLSH
on an average performs quite well, however it does fall below
the guaranteed recall value quite a few times. In summary,
our method provides rigorous guarantees of quality without
losing too much performance over BayesLSH.

\begin{figure*}[!htb]
\begin{center}
\includegraphics[width=140pt, height=20pt]{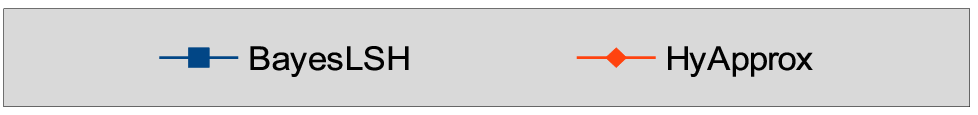}
\end{center}
\begin{center}
\hspace{-0.5in}
\subfigure[Twitter, Jaccard]
{
\includegraphics[width=120pt,height=100pt]{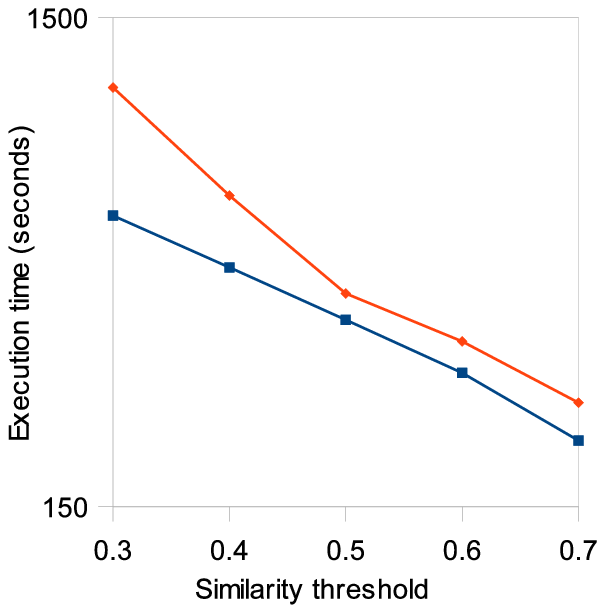}
\label{fig:twit-jac-lsh}
}
\subfigure[Twitter, Jaccard]
{
\includegraphics[width=120pt,height=100pt]{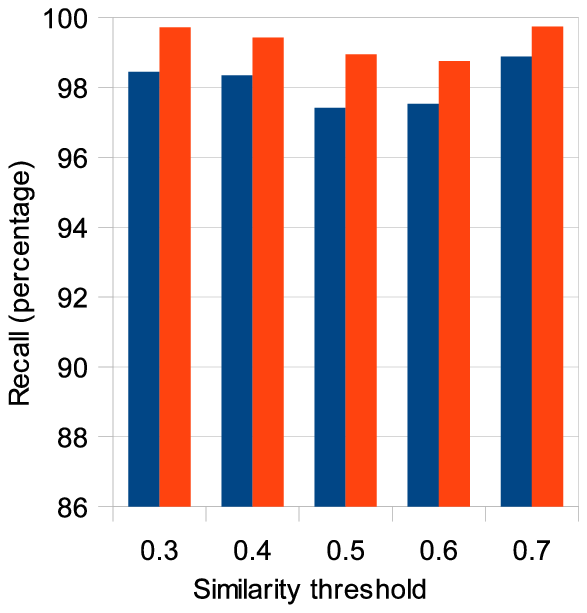}
\label{fig:twitter-jac-lsh-quality}
}
\subfigure[WikiWords100K,Jaccard]
{
\includegraphics[width=120pt,height=100pt]{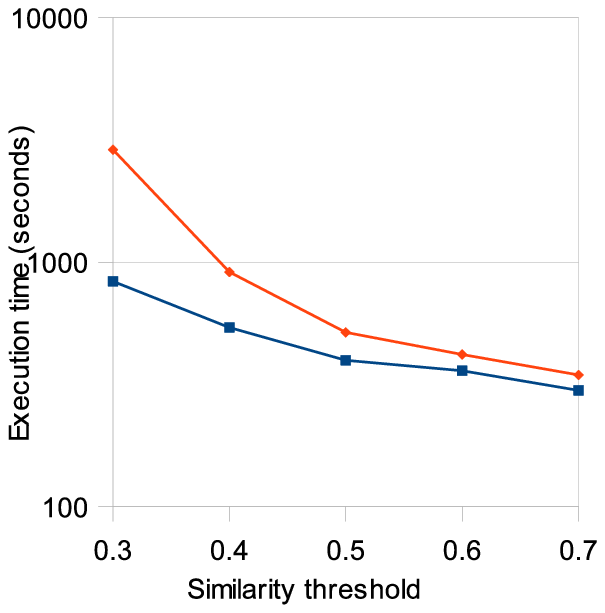}
\label{fig:wiki100-jac-lsh}
}
\subfigure[WikiWords100K,Jaccard]
{
\includegraphics[width=120pt,height=100pt]{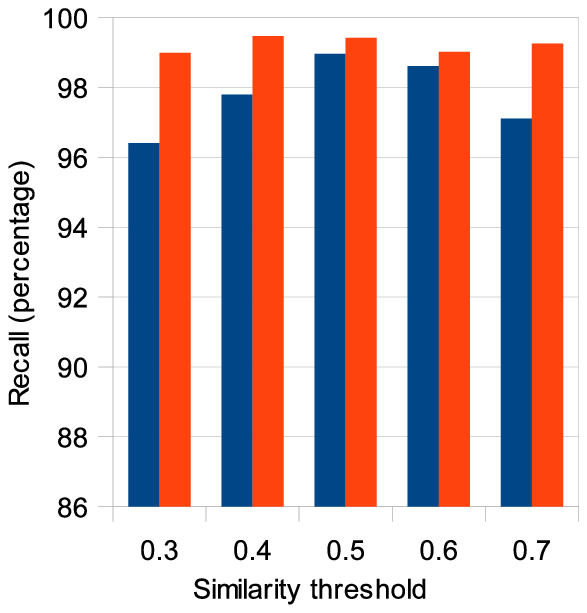}
\label{fig:wiki100-jac-lsh-quality}
}
\end{center}
\begin{center}
\hspace{-0.5in}
\subfigure[RCV,Jaccard]
{
\includegraphics[width=120pt,height=100pt]{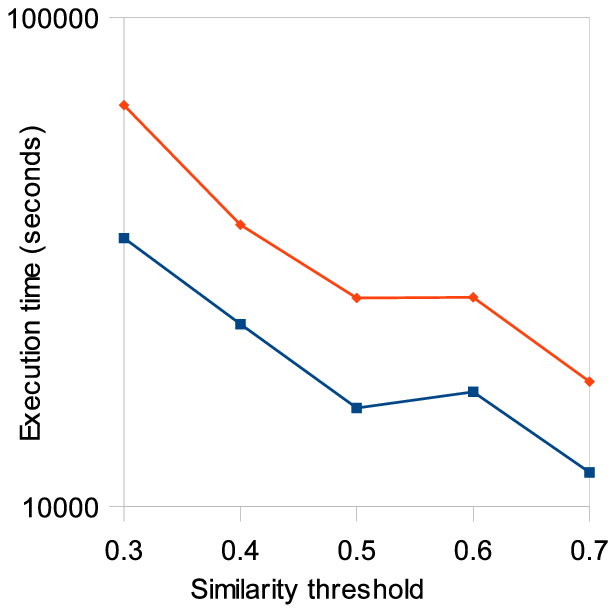}
\label{fig:rcv-jac-lsh}
}
\subfigure[RCV,Jaccard]
{
\includegraphics[width=120pt,height=100pt]{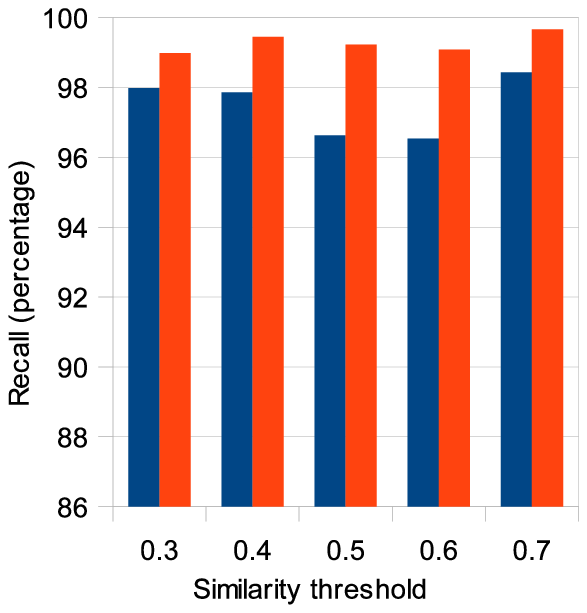}
\label{fig:rcv-jac-lsh-quality}
}
\subfigure[WikiLinks,Jaccard]
{
\includegraphics[width=120pt,height=100pt]{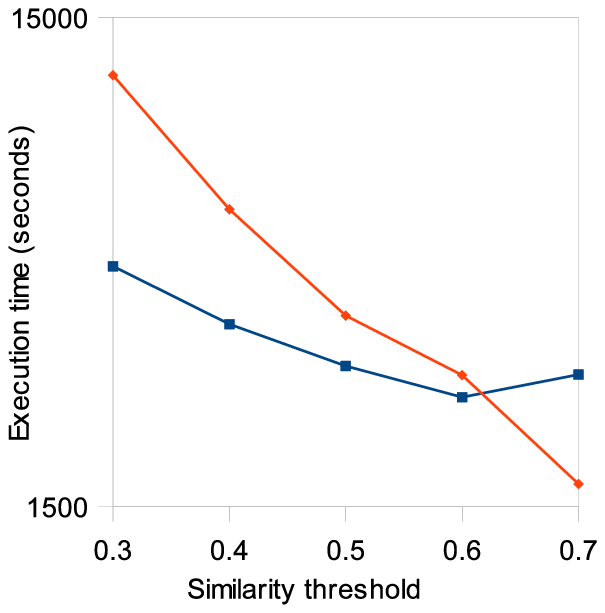}
\label{fig:wikilinks-jac-lsh}
}
\subfigure[WikiLinks,Jaccard]
{
\includegraphics[width=120pt,height=100pt]{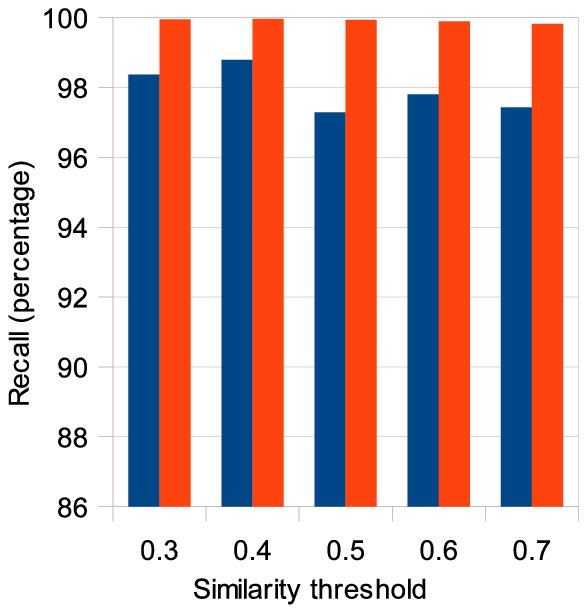}
\label{fig:wikilinks-jac-lsh-quality}
}

\end{center}

\begin{center}
\hspace{-0.5in}
\subfigure[Orkut,Jaccard]
{
\includegraphics[width=120pt,height=100pt]{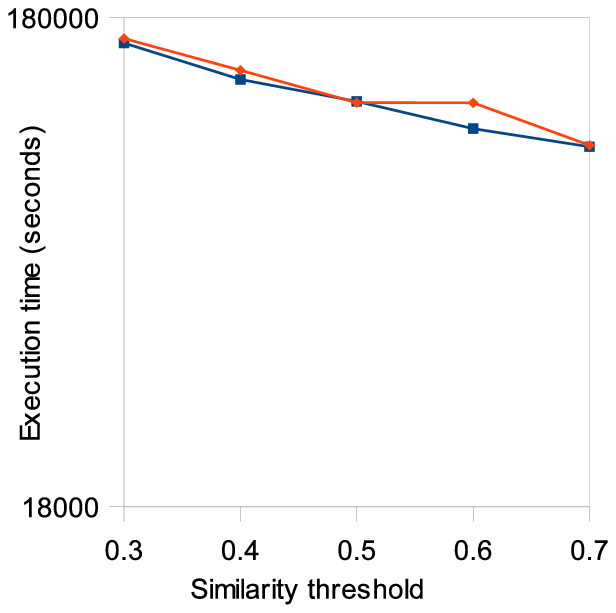}
\label{fig:orkut-jac-lsh}
}
\subfigure[Orkut,Jaccard]
{
\includegraphics[width=120pt,height=100pt]{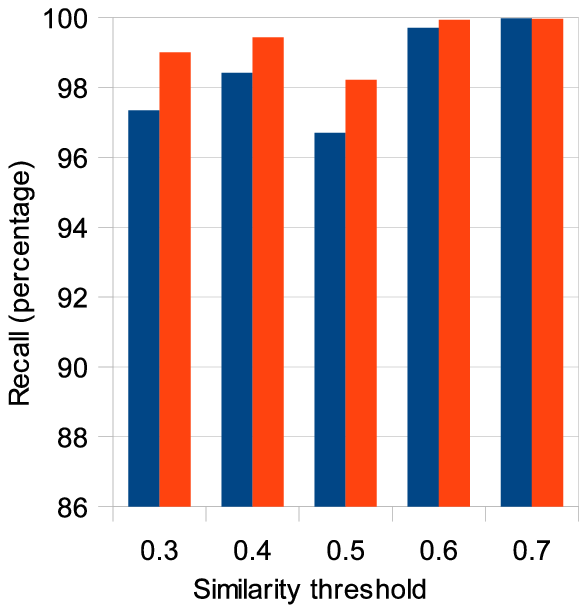}
\label{fig:orkut-jac-lsh-quality}
}
\subfigure[Twitter, cosine]
{
\includegraphics[width=120pt,height=100pt]{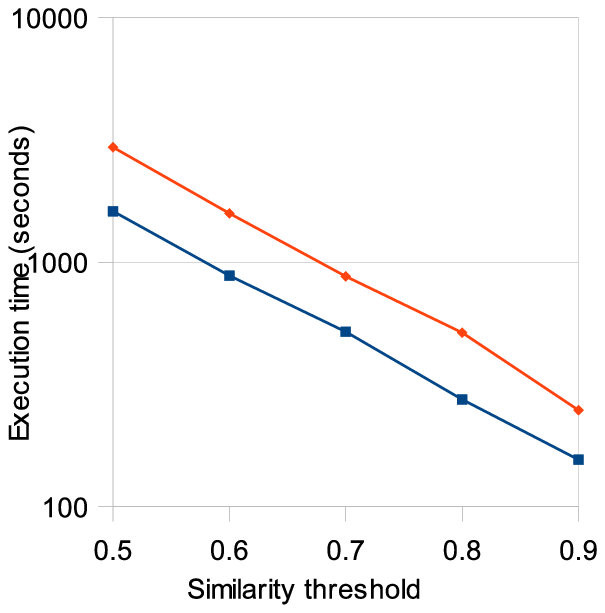}
\label{fig:twit-cos-lsh}
}
\subfigure[Twitter, cosine]
{
\includegraphics[width=120pt,height=100pt]{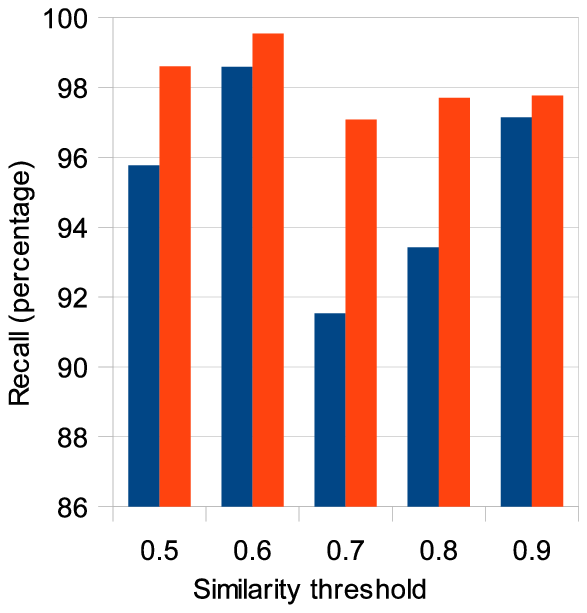}
\label{fig:twitter-cos-lsh-quality}
}
\end{center}

\begin{center}
\hspace{-0.5in}
\subfigure[WikiWords100K,cosine]
{
\includegraphics[width=120pt,height=100pt]{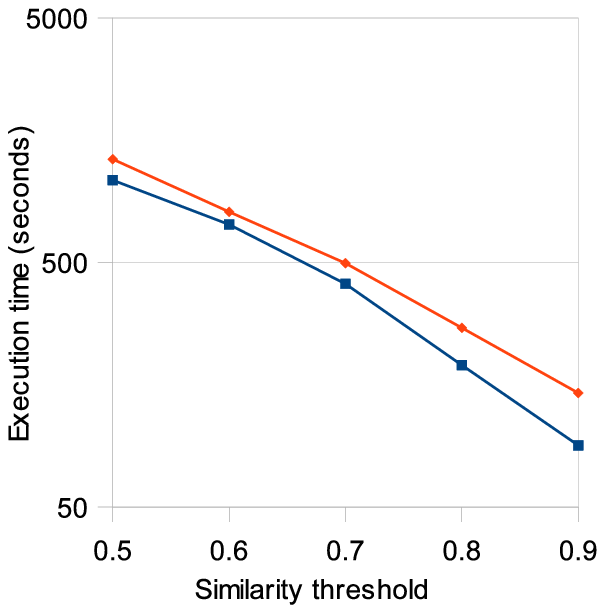}
\label{fig:wiki100-cos-lsh}
}
\subfigure[WikiWords100K,cosine]
{
\includegraphics[width=120pt,height=100pt]{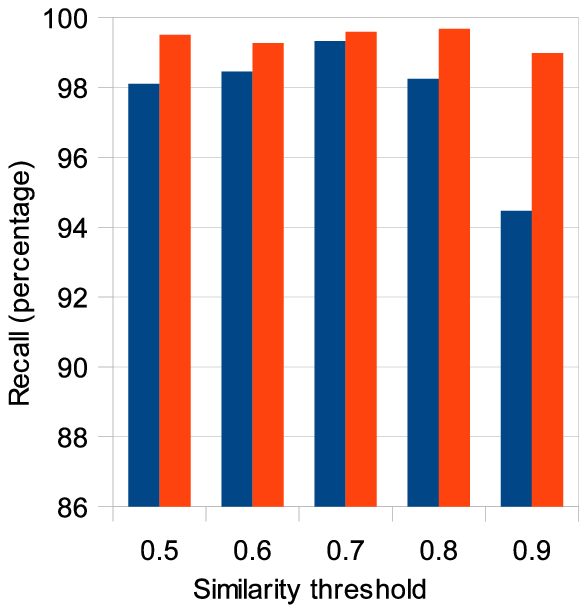}
\label{fig:wiki100-cos-lsh-quality}
}
\subfigure[RCV,cosine]
{
\includegraphics[width=120pt,height=100pt]{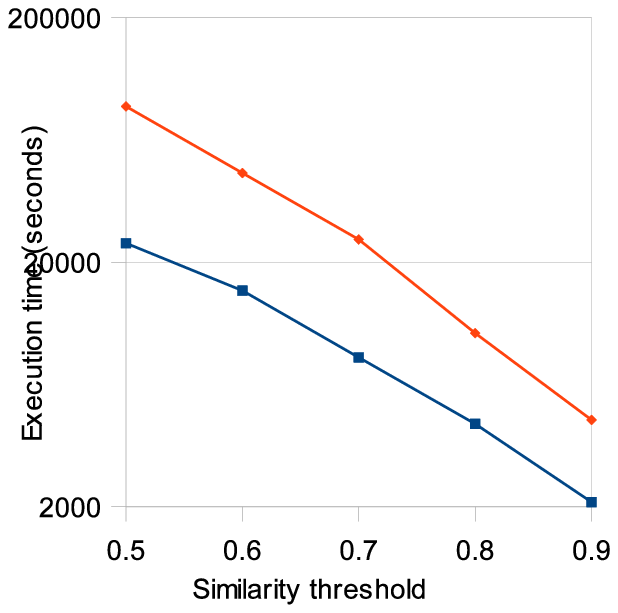}
\label{fig:rcv-cos-lsh}
}
\subfigure[RCV,cosine]
{
\includegraphics[width=120pt,height=100pt]{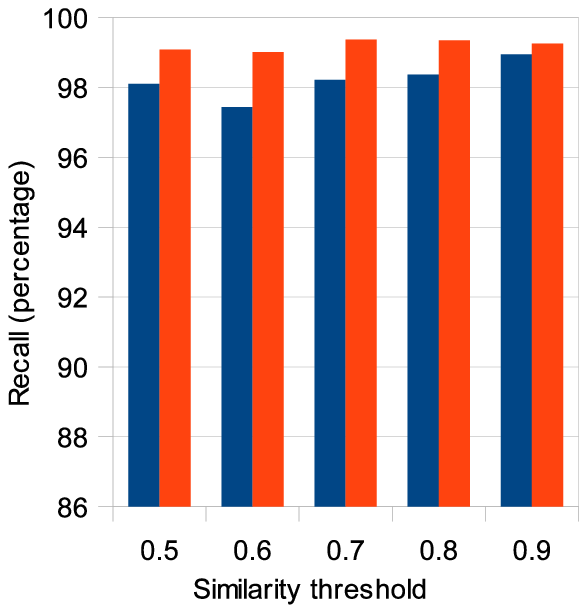}
\label{fig:rcv-cos-lsh-quality}
}
\end{center}
\begin{small}
\caption{Comparisons of algorithms with approximate similarity estimations.
}
\label{fig:perf-qual-lsh}
\end{small}
\end{figure*}

\section{Conclusions}
In this paper we propose 
principled approaches of doing all pairs
similarity search on a database of objects
with a given similarity measure. We describe
algorithms for handling two different scenarios -
i) the original data set is available and
the similarity of interest can be exactly computed from
the explicit representation of the data points and
ii) instead of the original dataset only a small
sketch of the data is available and similarity needs
to be approximately estimated. For both scenarios
we use LSH sketches (specific to the similarity
measure) of the data points. For the first case
we develop a fully principled approach of adaptively
comparing the hash sketches of a pair of points
and do composite hypothesis testing where the hypotheses
are similarity greater than or less than a threshold.
{\it Our key insight is a single test does not perform
well for all similarity values, hence we dynamically
choose a test for a candidate pair, based on a crude
estimate of the similarity of the pair}
. For the second case we additionally
 develop an adaptive algorithm for estimating the
approximate similarity between the pair. Our methods are
based on finding sequential fixed-width confidence
intervals. We compare our methods against state-of-the-art
allpairs similarity search algorithms BayesLSH/Lite that
does not precisely model the adaptive nature of the
problem. We also compare against the more traditional
sequential hypothesis testing technique -- SPRT. We
conclude that if quality guarantee is paramount,
then we need to use our sequential confidence interval
based techniques, and if performance is extremely important,
then BayesLSH/Lite is the obvious choice. Our Hybrid models
gives a very good tradeoff between the two extremes.
We show that our hybrid method always guarantees the minimum-prescribed
quality requirement (as specified by the input parameters), while being upto 2.1x faster than
SPRT and 8.8x faster than AllPairs. Our hybrid method is also improves the recall
by up to 5\% over BayesLSH/Lite, a contemporary state-of-the-art adaptive LSH approach.

\bibliographystyle{abbrv}
\bibliography{paper}

\end{document}